\documentclass[10pt, conference, compsocconf]{IEEEtran} 
\usepackage{times}

\usepackage{amsmath}
\usepackage{amssymb}
\usepackage{verbatim}
\usepackage{amsthm}
 \usepackage[dvips]{graphicx}
\newtheorem{theorem}{Theorem}[section]
 
\newtheorem{lemma}[theorem]{Lemma}

\theoremstyle{definition}

\newtheorem{definition}[theorem]{Definition}
\theoremstyle{remark}
\newtheorem{remark}[theorem]{Remark}

\newtheorem{example}{Example}[section]

\begin{document}

\title{Curse of Dimensionality in Pivot-based Indexes}
\author{\IEEEauthorblockN{Ilya Volnyansky, Vladimir Pestov}
\IEEEauthorblockA{
Department of Mathematics and Statistics\\
University of Ottawa\\ Ottawa, Ontario, Canada\\ 
\{ivoln012,vpest283\}@uottawa.ca\\
}
}

\maketitle
\thispagestyle{empty}

\begin{abstract}
%Working in the framework of statistical learning, 
We offer a theoretical validation of the curse of dimensionality in the pivot-based indexing of datasets for similarity search, by proving, in the framework of statistical learning, that in high dimensions no pivot-based indexing scheme can essentially outperform the linear scan.

A study of the asymptotic performance of pivot-based indexing schemes is performed on a sequence of datasets modeled as samples picked in i.i.d. fashion from a sequence of metric spaces.
We allow the size of the dataset to grow in relation to dimension, such that the dimension is superlogarithmic but subpolynomial in the size of the dataset.
The number of pivots is sublinear in the size of the dataset.
We pick the least restrictive cost model of similarity search where we count each distance calculation as a single computation and disregard the rest.

We demonstrate that if the intrinsic dimension of the spaces in the sense of concentration of measure phenomenon is linear in dimension, then
the performance of similarity search pivot-based indexes is asymptotically linear in the size of the dataset. 
%That is, for large enough $d$ the difference between using a an optimal pivot-based index and performing a search without an index at all is negligible.
% Thus we confirm the curse of dimensionality in this setting.

\end{abstract}
%--------------------------------------------------------------------------------------------------------------
\section{Introduction}
The problem of similarity search in databases is addressed by building indexing schemes of various types \cite{ciaccia:97,chavez:1,chavez:5,zezula}.
The goal of such structures is that a search algorithm can exploit them to perform similarity search in time sublinear in the database size.
That indexing schemes do not scale well with increasing dimension has been referred to as ``the curse of dimensionality'' \cite{beyer,indyk:1}.

We feel that in order to gain a better insight into the nature of the curse of dimensionality, it is necessary to have a precise mathematical understanding of the geometric and algorithmic aspects of what happens in genuinely high-dimensional datasets. With this purpose, we have chosen to analyse one of the most popular indexing schemes for similarity search, the one based on pivots \cite{bustos,chavez:1}. The mathematical setting for our analysis is a rigorous model of statistical learning theory \cite{blumer,devroye,vapnik,vidyasagar:03}, where datasets are drawn randomly from domains of increasing dimension.

This probabilistic setting is similar to that used in a previous asymptotic analysis of similarity search \cite{shaft:06}.
We also adopt a cost model where we count distance computations only, in line with \cite{shaft:06}.
Unlike this previous work, we let both the dimension $d$ and the size of the dataset $n$ grow as described in \cite{indyk:1}.
We also make the distinction between the dataset and the data space mathematically explicit.
In particular we emphasize that statements of the type ``all indexing scheme will degenerate to linear scan with increasing dimension'' (to paraphrase \cite{weber:98}) will always need to be qualified with estimates of the probability. For it is not impossible to sample a hypercube uniformly and come up with a ``distribution with a million clusters''\cite{shaft:06}.

Our analysis is done on a sequence of (data) spaces that exhibit the {\em concentration of measure} phenomenon \cite{gromov:83,milman:86} (Sect. \ref{s:conc}), a concept linked to what is called in \cite{shaft:06} {\em workloads with vanishing variance}. It is also in terms of this concentration of measure that we define the dimension $d$. To show that the above situation with a million clusters cannot happen (too often) we study the convergence of empirical probabilities to their true values using a result from Statistical Learning Theory \cite{vapnik}. We introduce a property of a sequence of spaces which is sufficient to invoke this result.

The conclusion of our analysis (Sect. \ref{s:main}) is that for high dimensional datasets the class of pivot-based indexing schemes cannot significantly outperform the baseline linear scan of checking every element of the database. 

%--------------------------------------------------------------------------------------------------------------
\section{Metrics, measures, and datasets}
We model the dataset as a sample of a metric space with measure. A {\em metric} (or: {\em distance}) on a set $X$ will be denoted by $\rho$, and we will not remind the definition.
% 
% \begin{definition}
% A {\em metric space} is a set $\Omega$ equipped with function $\rho: \Omega\times\Omega\rightarrow\mathbb{R}$ s.t.
% \begin{itemize}
% \item $\rho(\omega_1, \omega_2)\geqslant 0$, $=0$ if and only if $\omega_1=\omega_2$
% \item $\rho(\omega_1, \omega_2)=\rho(\omega_2, \omega_1)$ and
% \item $\rho(\omega_1, \omega_2)\leqslant \rho(\omega_1, \omega_3)+\rho(\omega_3, \omega_2)$
% \end{itemize}
% \end{definition}
% 
% The function $\rho$ is a {\em distance} (also: {\em metric} ), the use of which gives rise to the notion of balls:
% \begin{definition}
The (open) ball of radius $r$ and centre $q$ in a metric space $(\Omega, \rho )$ is denoted 
\[B_{r}(q):=\{\omega\in\Omega\, |\, \rho(q,\omega )< 
r\}.\]
% \end{definition}

The 
%In addition,
%\begin{definition}
%The 
family $\mathcal{B}=\mathcal{B}_\Omega$ of {\em Borel subsets} of a metric space $(\Omega, \rho )$ is the smallest family containing all the open balls and the entire set $\Omega$ and closed under complements and countable unions. 
%We will denote this family by $\mathcal{B}$

%\end{definition} 
%Then, we can endow the metric space with a probability measure:
%\begin{definition}
A {\em (Borel) probability measure} on the space $(\Omega, \rho )$ is a function $\mu: \mathcal{B}_{\Omega}\rightarrow [0,1]$ s.t.
$\mu(\Omega)=1$, and which is countably additive: for a sequence $B_1,B_2\ldots,$ of pairwise disjoint sets from $\mathcal{B}$, $\mu(\bigcup_i{B_i})=\sum_i{\mu(B_i)}$
%\end{definition}

A dataset however large is always a finite subset $X\subset\Omega$.
It naturally inherits the metric $\rho|_X$ and in place of $\mu$ supports the {\em normalized counting} (also: {\em empirical}) probability measure:
 \[\mu_\# (A) = \frac{|A\cap X|}{|X|}.\]
% (which is indeed a probability measure).
We will treat $X$ as a sample of $\Omega$ with regard to the measure $\mu$, that is, a sequence of i.i.d. random variables $(X_i)\sim\mu$.

Given a domain together with a dataset $X\subseteq\Omega$, we can perform several kinds of similarity queries, with our focus on two main ones.
A {\em $k$ nearest neighbour query} consists of, given query centre (key) $q\in\Omega$, finding the k closest elements in $X$ to $q$.
To answer a {\em range query} is to find all the elements in $X$ within distance $r$ from $q$.

To distinguish between $X$ and $\Omega$ formally is necessary precisely because a typical search will begin with a centre $q\in\Omega$, with $q\in X$ as well being rare.

To answer a similarity query we can revert to the strategy of looking up every element in $x\in X$ and calculating $\rho(q,x)$.
Following \cite{chavez:1}, we adopt the number of distance calculations $\rho(q,x)$ as the unit of time complexity.
In that framework we will call the above strategy a {\em linear scan.} 
% as it is clearly linear in the size of the dataset.

% \begin{definition} 
An {\em indexing scheme} is a structure whose aim is to speed up the execution of similarity queries on a particular dataset, typically consisting of some pre-calculated values and an algorithm.
% \end{definition}

% Dozens of ideas of how to build indexes were presented and new ones, having various advantages and analyzed in different ways, are invented continuously.

%Attempts at categorizing and providing a unified framework are recent: a book on the subject appeared in 2005 \cite{zezula} and the first international conference on similarity search \cite{sisap} was held in 2008.
%Given that the triangle inequality is essentially the only way of constructing indexes, the diversity of the various indexing schemes is astonishing.
% There are tree and flat structures, structures that try to optimize inserts and deletes, trees that are deep or shallow, balanced or not, with claimed
% computational complexities from constant to exponential in $n =$ size of $X$.
%Furthermore because search algorithms is a topic of research in multiple disciplines, including for example pattern recognition \cite{devroye} where a ``nearest neighbour'' algorithm is almost equivalent to nearest neighbour search, the same solutions have been reinvented multiple times and called different names \cite{chavez:1}, \cite{clarkson:05}.

\section{The curse of dimensionality}
An often repeated observation is the inability of many existing algorithms to deal with high dimensional datasets (e.g. \cite{beyer})-- a phenomenon described as the {\em curse of dimensionality}, when 
% Simply put, when algorithms are run on Euclidian datasets of increasing dimension, 
performance drops exponentially as a function of dimension.

The concept of dimension in a general metric space with measure is less precise.
Clearly it has to obey our intuition in Euclidian space so for example a plane in the 10-dimensional space $\mathbb{R}^n$ is still 2-dimensional, and
it would be desirable for a uniformly distributed ball in $\mathbb{R}^d$ to be $d$-dimensional.
% 
% One approach is to focus on the metric space properties of $\Omega$, and take advantage of already known concepts for metric spaces, such as packing numbers and $\epsilon$-nets. This leads to the {\em Assouad dimension} \cite{clarkson:05}.
% Some results to the effect that small (Assouad) dimension implies fast index are stated in \cite{clarkson:05}.

A version of intrinsic dimension was proposed by the
%Perhaps a more fundamental concept called {\em intrinsic dimension} by the 
authors \cite{chavez:2} as
 \[\tilde{d} = \frac{\mathrm{E}(\rho(x,y))^2}{2\mathrm{Var}(\rho(x,y))},\]
 where $x,y\sim \mu$, the distribution of points in $\Omega$.
It is based on the observation that if the histogram of distances from $q$ to points in $X$ shows a lot of ``concentration'', this will be a hard query to process as it is harder to rule out points using a triangle inequality type approach.
% As the numerator is a just a normalizing term, it is best to think of intrinsic dimension as a proxy for the variance of a histogram of distances for an average $q$.
% In fact we will opt to normalize spaces so that the average distance between two points does not increase with dimension, leaving the numerator a constant.
That the above dimension is asymptotically equal to the usual notion in Euclidian spaces is mentioned in \cite{chavez:2}, where a result on time complexity of search in term of $\tilde{d}$ is also stated.
It is a lower bound on the order of $\tilde{d}\, \ln (n)$. 

In this article, we will use another approach to the intrinsic dimension, elaborated in \cite{pestov:08} and also based on the phenomenon of concentration of measure, cf. Sect. \ref{s:conc}.

In general the time complexity we are looking for in search depends both on dimension (henceforth we will simply call it $d$) and size of dataset $n$.
% The need to take into account $d$ arises from the fact that for fixed (small) $d$ we can find indexing schemes that work \cite{chavez:1}, yet they fail on datasets of similar size but of different dimension, hence the curse of dimensionality. 
An asymptotic analysis of the performance of indexing schemes will therefore involve both
$d\rightarrow\infty$ and $n\rightarrow\infty$.
Search in sublinear time in $n$ is an obvious goal:
\[\text{querytime }=o(n).\]
where by querytime we mean the average time it takes for a similarity query to execute, time measured in distance computations.
% The average here is computed over a reasonable space of possible queries, on which we will touch later.

Storage is also important, with at most polynomial storage allowed in theoretical analysis (though in practice even $n^2$ may be too much):
\[\text{storage }=n^{O(1)}.\]
For the pivot-based indexing scheme the storage will be measured by the number of distances {\em stored}.

We will follow an approach in the authoritative survey by \cite{indyk:1} and focus on a particular range for rate of growth for dimension $d$, superlogarithmic but subpolynomial in $n$:
 \begin{equation}d=\omega(\log n)\label{eq:lbN}\end{equation}
 \begin{equation}d=n^{o(1)}\end{equation}
%make a reference to hamming cubes chapter
This choice of bounds is due to a case study of the Hamming cubes. Recall that 
%\begin{definition}[The Hamming Cubes $\Sigma ^d$]\label{def:ham}
%thin out the example, now have it in chapter 1 as well
% The 
the {\em Hamming cube} $\Sigma ^d$ of dimension $d$ is the set of all binary sequences of length $d$, 
%that is its elements are of the form
% \[\boldsymbol{x}=(0,1,1,0,1,\ldots,1)\]
and the distance between two strings is just the number of elements they don't have in common divided by $d$:
\[\rho(\boldsymbol{x},\boldsymbol{y})=\frac{\sum_{i=1}^{d}{|x_i-y_i|}}{d}\]
(the {\em normalized Hamming distance}).
% (When considering the Hamming cube as a metric space with measure, we will give it a uniform distribution.)
%\end{definition}

In the case where $d$ grows slowly, $d=O(\log n)$,
%the entire space $\Omega$ is so small relative to the size of the dataset that 
all possible queries can be pre-computed and stored without breaking the polynomial storage requirement. Hence the lower bound.
% The size of $\Omega$ is just $2^d$ which becomes on the order of $n$ for sublogarithmic $d$, and so
% As there are on the order of $n$ possible radii, there are only $n^2$ possible queries which can be all precomputed.
%So to build a general framework for asymptotic bounds it seems necessary that $d$ grow strictly faster than $\log n$
%make a ref. to more lower bounds
% As we consider algorithms that are exponential in $d$ to suffer from the curse of dimensionality, we will require querytime polynomial in $d$ (\cite{indyk:1}):
% \[\text{querytime }=d^{O(1)}\]
The upper bound results from the observation that if $d$ grew so fast that $n=d^{O(1)}$, a sequential scan would be polynomial in $d$ and so acceptable. 
% fast.
% As nothing needs to be proven in that case, we focus on when $d$ is subpolynomial in $n$ and require an algorithm polynomial in $d$ and hence subpolynomial in $n$.
%We will adopt the view that these bounds on $d$ are a reasonable setting for the investigation of performance of various index based query algorithms. 
%While $d$ grows fast enough to not render the problem trivial, we disregard high rates of growth for which proven examples of the ``curse'' already exist.

Summarizing:
{\em The goal of finding a scalable index is to find polynomial (preferably degree less than 2) $n$ storage algorithm that allows search in polynomial $d$ time.} 

This stands in contrast to the {\em curse of dimensionality conjecture}, as stated in \cite{indyk:1}:

{\em
If $d=\omega(\log n)$ and $d=n^{o(1)}$, any sequence of indexes built on a sequence of datasets $X_d\subset\Sigma_d$ allowing exact nearest neighbour search in time polynomial in $d$ must use $n^{\omega(1)}$ space.
}

The conjecture remains unproven in the case of general indexing schemes.
The goal of this article is to show that at least for pivot-based indexes the above conjecture holds even in a strengthened form.

\section{Pivot-based indexing}
We will focus on one class of indexing schemes, the pivot-based index (e.g. AESA, MVPT, BKT,...see \cite{chavez:1} and \cite{zezula}).
% sharing essentially one method of executing a search.
The index is built using a set of {\em pivots} $\{p_1\ldots p_k\}$ from $\Omega$, and consists of the array of $n\times k$ distances 
 \[\rho(x,p_i)\text{, }1\leqslant i\leqslant k\text{, }x\in X.\]
Given a range query with radius $r$ and centre $q$, the $k$ distances $\rho(q,p_1)\ldots \rho(q,p_k)$ are computed so that $\rho(q,x)$ can be lower-bounded by the triangle inequality:
% \[\rho (q,x)\geqslant|\rho (q,p_i)-\rho (x,p_i)|.\]
% Since this happens for any $i$, we can establish:
\[ \rho (q,x)\geqslant\sup_{1\leqslant i\leqslant k}|\rho (q,p_i)-\rho (x,p_i)|. \]
It is useful to think of a new distance function,
% based on the $k$ pivots:
\[\rho_k (q,x):= \sup_{1\leqslant i\leqslant k}|\rho (q,p_i)-\rho (x,p_i)|.\]

The fact that $\rho (q,x)\geqslant \rho_k (q,x)$ can be used to discard all $x$ satisfying $\rho_k (q,x) > r$. For the remaining points,
% Only if this inequality is not satisfied, 
the algorithm will 
% perform the
% Therefore the algorithm consists of checking this inequaity, and if it is not satisfied, performing the 
% (possibly expensive) distance calculation to 
verify if  
$ \rho (q,x) \leq r$. If it is true, the point is returned.

We will only analyze {\em range} queries;
% with pivot-based algorithms:
% chiefly because they are easier to execute.
$k$-nearest neighbour queries can always be simulated by a range query of suitable radius 
%with the radius set to the distance to the $k$th neighbour 
\cite{zezula}.

For a query centre $q$ denote 
by $C_q$ all the points of $X$ satifying $\rho_k (q,x) > r$, i.e. all the elements to be discarded.
Making $C_q$ large is the primary way of cutting the cost of search in our cost model. 
Of course we can achieve this trivially with a very large number of pivots. This will defeat the purpose however as
\[\text{Cost of range search }= k+|X \backslash C_q|\]

The most often used solution is to keep adding pivots as long as it is found experimentally to decrease the cost of search.
If $k$ is small, on the order of $\log n$ (as often space limitations require),
the most important component of cost becomes the size of $X\\C_q$ and this is where the choice of pivots would seem to matter.
 Various approaches to pivot selection have been investigated in \cite{bustos}.
 The empirical results seem to suggest that a moderate reduction in the number of distance 
 computations can be achieved, although the relative improvement drops with increasing dimension.
 
\begin{remark}[The number of pivots $k$.]
There are indexing schemes, like AESA \cite{zezula} where $k=n$.
However, in many situations $n^2$ storage is not practical, and it has even been argued that under certain assumptions the optimal number of pivots is on the order of $\ln n$ \cite{chavez:1}. It is also true that the query algorithm we analyze has complexity at least $k$ so only schemes with $k=o(n)$ can claim to beat the curse of dimensionality.
%: in a database consisting of an enormous amount of records is there really space for storing an index that is a {\em square} of that very large number? (keeping Google and the Internet as a case in point).
% and that even this number is rarely reached in practice 
%Therefore it would not be unrealistic to restrict the analysis to pivot size much smaller than $n$.
%For example we could require
%\[k=O(\log n)\]
%We also have to keep in mind that the query algorithm we have described requires at least $k$ distance computation so if $k=n$ the query is linear in $n$ which for our purposes is no better than a linear search.
\end{remark}
 
%To recap, the goal of constructing a pivot index is to select pivots (in particular $k$) that will result in $C$ being large most of the time, so that most queries execute rapidly.
%We will show situations where this is not possible.
%--------------------------------------------------------------------------------------------------------------
\section{\label{s:conc}Concentration of measure}
Perhaps the most compelling way to describe the concentration of measure phenomenon is to draw a picture.
We will attempt to draw the (surface of the) unit sphere $\mathbb{S}^d$ for various $d$, by sampling points and projecting them onto a flat surface.
Any orthogonal projection, say taking the first 2 coordinates, will give the picture similar to that in Figure \ref{figure:spheres}.
Under the sampling approach, it appears that high dimensional spheres are ``small'' even if we know their diameter to be a constant irrespective of $d$.

\begin{figure}[h]
\centering
		\includegraphics[height=1.5in,width=1.5in]{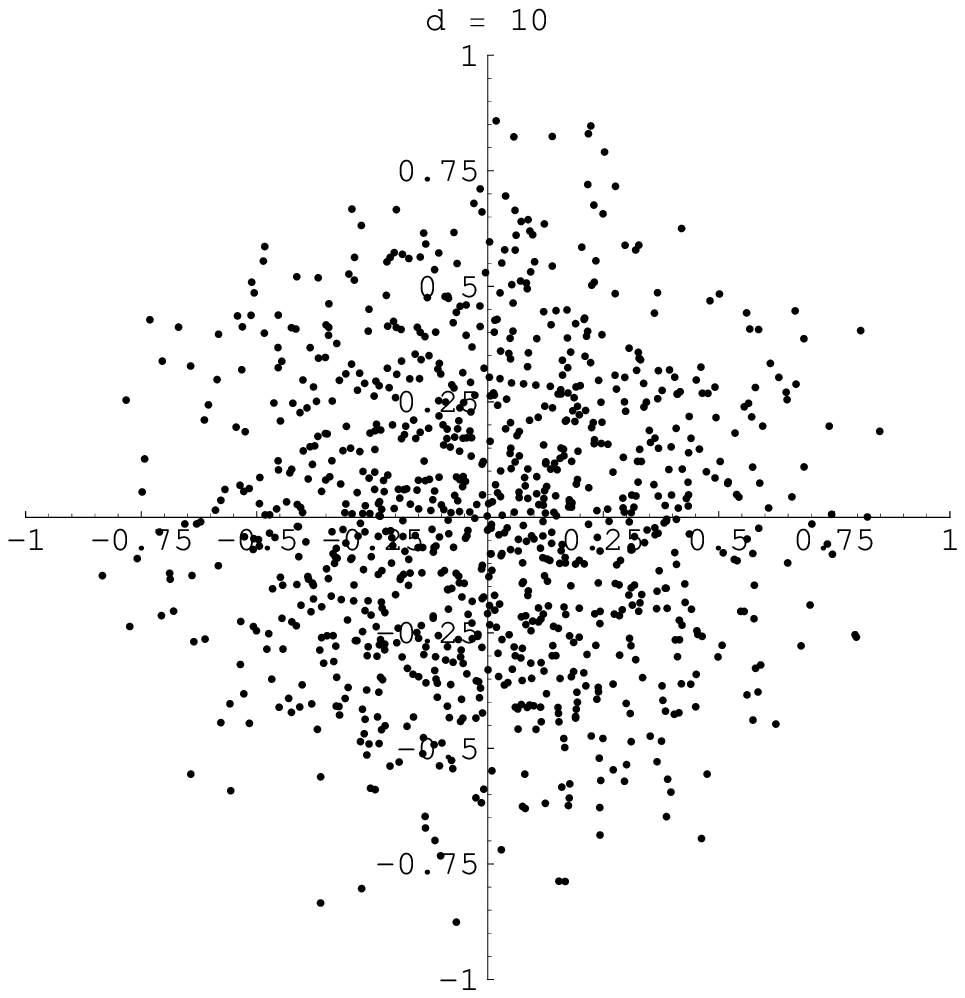}
		\hfill
		\includegraphics[height=1.5in,width=1.5in]{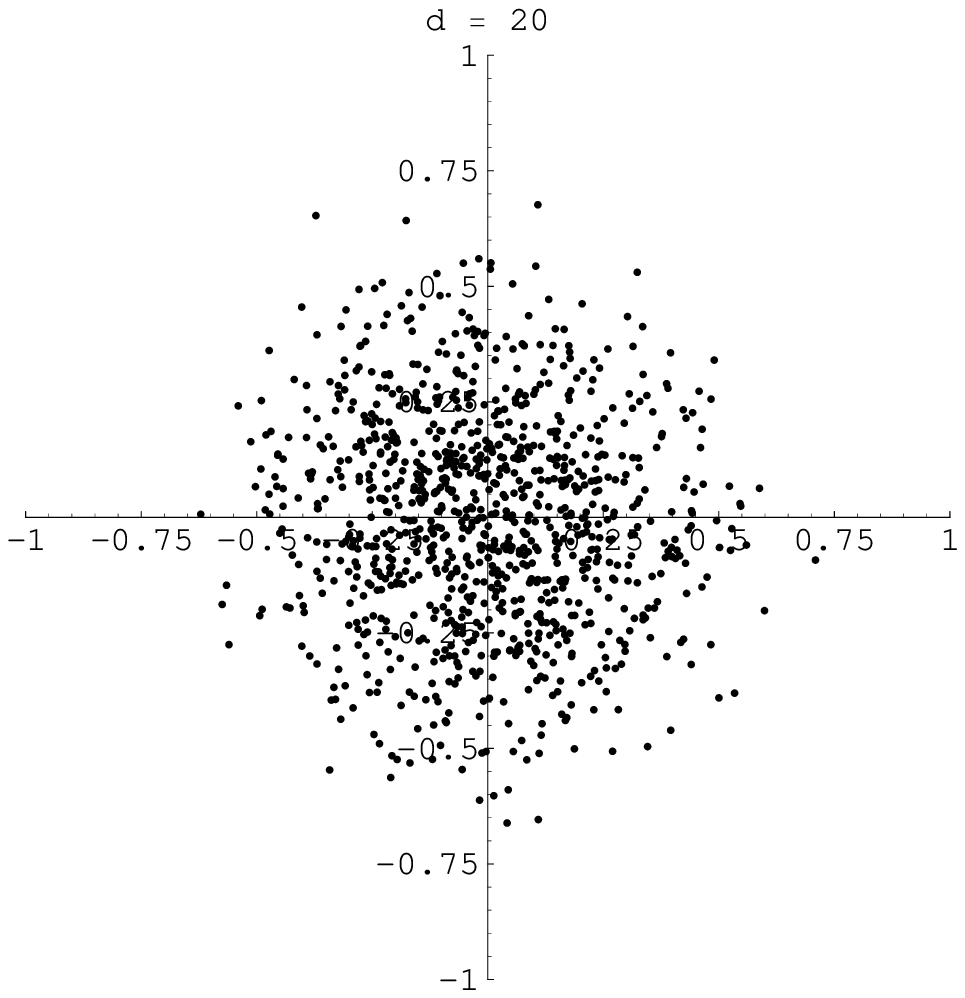}
		\hfill
		
		\vspace{2ex}
		
		\includegraphics[height=1.5in,width=1.5in]{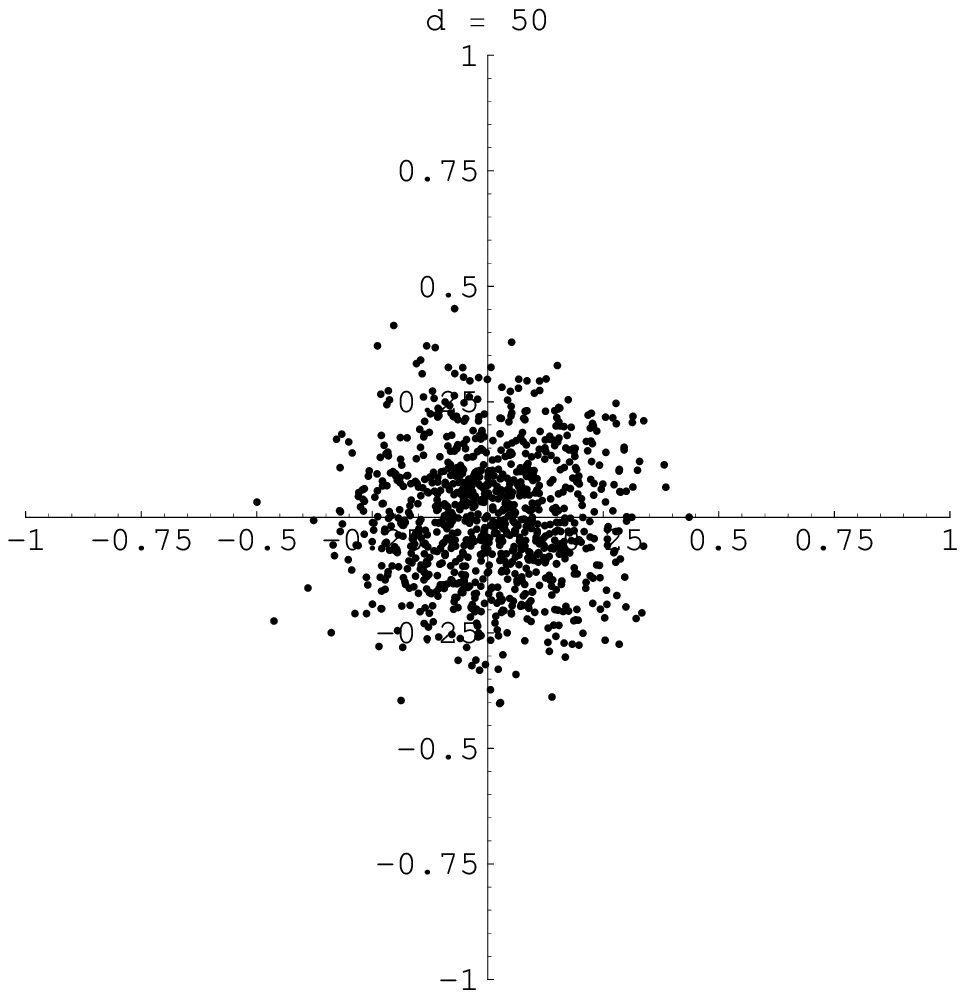}
		\hfill
		\includegraphics[height=1.5in,width=1.5in]{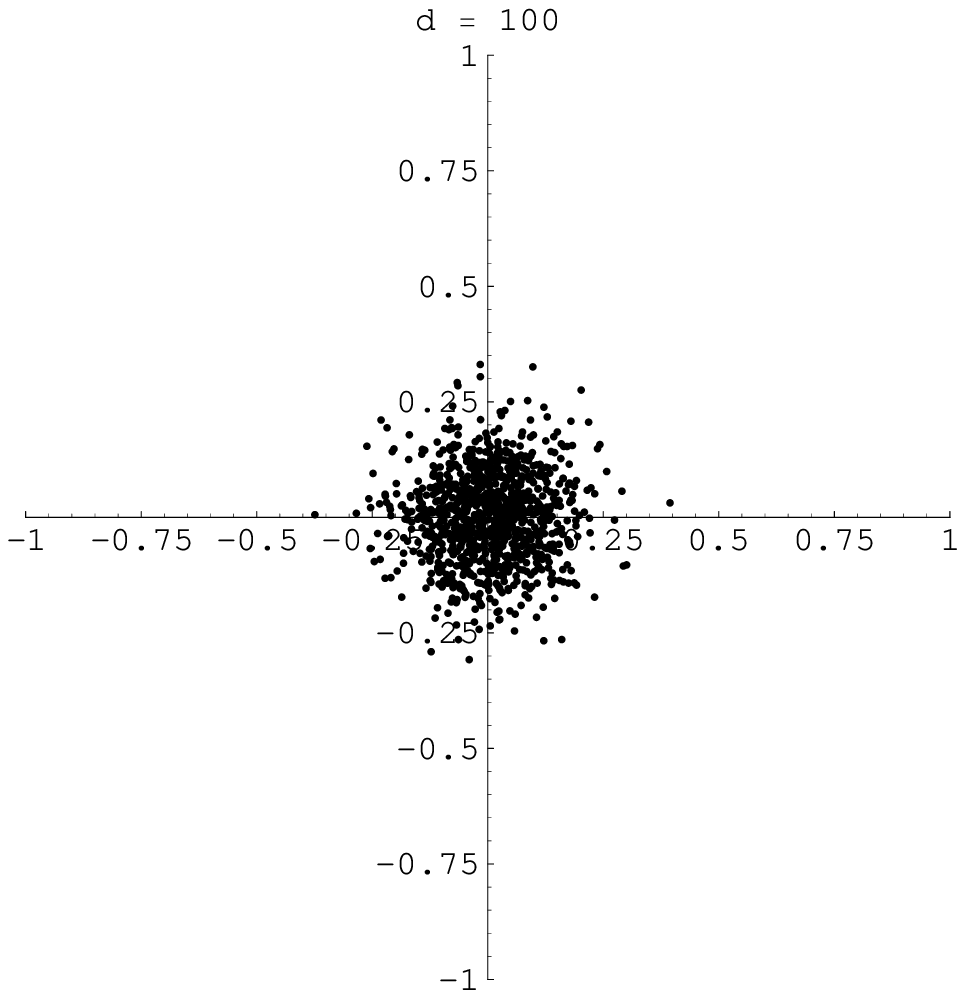}
			\caption{Projection of randomly sampled spheres of various dimensions d=10, 20, 50, 100}
	\label{figure:spheres}
\end{figure}

% Of course there is a need to formalize the concepts of ``small'' and ``orthodox projection''.

This phenomenon is observed in a much greater variety of situations and formalized as follows.
%\begin{definition}
Given a metric space $(\Omega ,\rho)$, define the $\epsilon$-neighborhood $A_{\epsilon}$ of $A\subset\Omega$ as
\[A_{\epsilon}=\{\omega\in\Omega|\rho(\omega,a)<\epsilon\text{ for some } a\in A\}.\]
%\end{definition}

% We want to define a function $\alpha$ s.t. if $\mu(A)\geqslant 1/2$ then
% \[\mu(A_{\epsilon})\geqslant 1-\alpha ( \epsilon )\]
% 
% In a sense we will pick the best such $\alpha$ and call it the {\em  concentration function}:

\begin{definition}
The {\em concentration function}  $\alpha = \alpha_{\Omega}$ of
a metric space with measure $(\Omega ,\rho, \mu)$ is defined as
\begin{align*}
\alpha(0)&=1/2,\\
\alpha(\epsilon)& = \sup\{1-\mu(A_{\epsilon})|A\subset\Omega,\mu(A)\geqslant\tfrac{1}{2}\}\quad ,\,\epsilon > 0.\\
\end{align*}
\end{definition}

To put it less formally, we are trying to measure how much of the space remains after ``fat'' is added to a
 somewhat large set in the form of an $\epsilon$ neighborhood. When very little remains, we say that the concentration of measure takes place.
% Making the concept of ``little'' more precise, {\em normal }concentration of measure is considered to be taking place when $C$, $c>0$ exist such that
% \[\alpha(\epsilon)<C\text{e}^{-c d \epsilon ^2},\]
% where $d$ is the (intrinsic) dimension parameter.

\begin{example}{The spheres $\mathbb{S}^d$ in $\mathbb{R}^{d+1}$,}
taken with the geodesic or Euclidian distance and the normalized invariant measure, produce a concentration function bounded as follows \cite{milman:86}:
\[\alpha_n(\epsilon)\leqslant\text{e}^{-(d-1)\epsilon^2/2}.\]
\begin{figure}[h]
\centering
		\includegraphics[height=1in,width=2in]{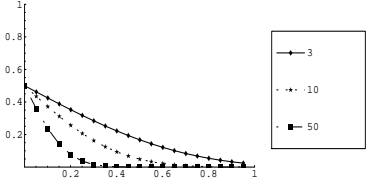}
	\caption{The concentration functions of various spheres}
	\label{figure:spheresConc}
\end{figure}
In this case an exact expression for the concentration function is known \cite{milman:86}, based on the fact that the half-sphere, among all subsets of measure at least $1/2$, will always produce the smallest $\epsilon$-neighborhood, no matter the $\epsilon$.
% The measure of this neighborhood is given by
% \[\left(\int_{-\pi/2}^{\epsilon}{ \cos^{d-2}x \text{d}x} \right)\Big/\left(\int_{-\pi/2}^{\pi/2}{ \cos^{d-2}x \text{d}x} \right).\]
% An estimation of this value can be arrived at via numeric integration.
A plot of the resulting concentration functions, for several values of $d$, appears in Figure \ref{figure:spheresConc}.
\end{example}

% This example is particularly interesting as increasing dimension leads to increased concentration of measure phenomenon.
\begin{definition} 
A sequence of spaces $(\Omega_d )_{d=1}^{\infty}$ is a {\em normal L\'evy family} \cite{milman:86} if $C$, $c>0$ exist such that
\[\alpha(\epsilon)<C\text{e}^{-c\epsilon ^2d}.\]
% Thus it the same notion of what is a tight concentration function as above.
\end{definition}

\begin{example}{The Balls $\mathbb{B}^d$,}
taken with the Euclidian distance and the uniform probability measure ($d$-dimensional Lebesgue), form a normal L\'evy family.
\end{example}

\begin{example}{The Hamming Cubes $\Sigma ^d$}
form a normal L\'evy family under the normalized Hamming metric and the uniform measure.
\end{example}

The concentration of measure can be equivalently described in terms of Lipschitz functions. 
%
% \begin{definition}
Recalling that a function $f:\Omega\rightarrow\mathbb{R}$ is {\em 1-Lipschitz }if 
\[\forall x,y\in\Omega,\,\, |f(x)-f(y)|\leqslant \rho (x,y).\]
% In general a function $f$ is p-Lipschitz if for all $x$ and $y$, $|f(x)-f(y)|\leqslant p\rho (x,y)$.
% \end{definition}
% We note that $\rho (q,.)$ the function assigning to $\omega$ its distance to $q$ is 1-Lipschitz.
%In spaces that have a tight concentration $\alpha$, Lipschitz functions will be nearly constant, and one candidate for this constant is a median value.

Recalling further that 
% \begin{definition}
%\noindent A 
a {\em median} of function $f:(\Omega,\rho,\mu)\rightarrow\mathbb{R}$ is any number $M$ satisfying:
\[\mu\{\omega | f(\omega)\leqslant M\}\geqslant 1/2\text{ and }  \mu\{\omega | f(\omega) \geqslant M \}\geqslant 1/2.\]
%This is a slight generalization of the usual concept of median of a set of numbers, only no attempt is made to make it unique. 
%Discrete functions may very well have multiple valid values for $M$.
% \end{definition}

It is then relatively straightforward to prove:

\begin{theorem}[Cf. \cite{milman:86}]
\label{thm:lip}
 For a 1-Lipschitz function f defined on space $(\Omega,\mu,\rho)$:
 \[\forall\epsilon > 0,\quad \mu\{\omega |\ |f(\omega)-M|>\epsilon\} < 2\alpha(\epsilon).\]
\end{theorem}%the proof?

The relevance of concentration of measure in indexing is noted in \cite{pestov:00}. Observe that
% It relies on the observation that  
\[\rho (\cdot,p):\Omega\rightarrow\mathbb{R}:\omega\mapsto \rho (\omega,p)\]
is 1-Lipschitz for any $p$ and in particular a pivot. Hence Theorem \ref{thm:lip} can be applied to obtain a bound on the deviation from the median $M=M_p$ of function $\rho (\cdot,p)$:
\[\forall r > 0,\quad \mu\{\omega |\ |\rho (\omega,p)-M|>r\} < 2\alpha(r).\]
% since $p$ is a general element of $\Omega$ the statement holds individually for each pivot $p_i$. 
We combine these statements for all pivots $p_i$:
\begin{equation*}
%\label{median:1}
\forall r > 0,\quad \mu\{\omega |\sup_{1\leqslant i\leqslant k}|\rho (\omega,p_i)-M_i|>r / 2\} < 2k\alpha(r / 2),
\end{equation*}
as the probability of the union can always be upperbounded by the sum of the probabilities.
We note that no assumptions about independence are used: the sequence $(p_i)$ can be chosen in any way.
%We used $r/2$ so as to get rid of the $M_i$:
Next, for all query centres $q$ except a set of measure $<1-\alpha(r/2)$:
\[\forall r > 0,\quad \mu\{\omega |\sup_{1\leqslant i\leqslant k}|\rho (\omega,p_i)-\rho (q,p_i)|>r\} < 2k\alpha(r / 2).\]
% Comparing this to the definition of $C$,
%\[C=\{x \vert \rho _k(q,x) > r \}\]
% it is apparent that the only difference between the set we are upperbounding and $C$ is that one is defined over all of $\Omega$ and the other, just for $X$.
We could introduce a set 
\[\mathcal{C}_q=\{\omega \vert \rho _k(q,\omega) > r \}\]
and think of $C_q$ as the observation of $\mathcal{C}_q$ under $\mu_{\# }$. To recap: for a randomly chosen query centre and each query radius $r>0$, with probability $>1-\alpha(r/2)$,
%
% To restate the upperbound in terms of $\mathcal{C}$, 
\begin{equation}\label{eq:conc}
%\forall r > 0,\quad 
\mu (\mathcal{C}_q )< 2k\alpha(r / 2).
\end{equation}
%So in effect, assuming that $2k\alpha(r / 2)$ is small, we have that $\mathcal{C}$ is small as well.
%-------------------------------------------------------------------------------------------------------
\begin{remark}
We point out that Theorem \ref{thm:lip} applied to the distance function $\rho$ gives a bound on the variance of $\rho(\cdot,p)$.
This, together with a ``uniformity of view'' type assumption as in \cite{chavez:2} leads us to conclude that the variance of $\rho(\cdot,\cdot)$ converges to zero in L\'evy families.
This argument can be formalized to demonstrate the connection to the assumption of vanishing variance on the sequence of data spaces made in \cite{shaft:06}.
In our view that assumption is just a variation on concentration of measure.
The differences lie in certain technical details, like the division by expectation of $\rho(\cdot,\cdot)$ in \cite{shaft:06}.
Here we simply avoid the issue by normalizing spaces so that the expectation of $\rho(\cdot,\cdot)$ tends to a constant.
This normalization also fixes the problem of distance to nearest neigbour (e.g. \cite{weber:98}) as we demonstrate in the next section.

\end{remark}
\section{Radius of queries in L\'evy families}\label{sec:radius}
In our asymptotic analysis, 
% We have described above how 
we would like to normalize spaces so that the median distance between two points stays about the same.
Here we will extract consequences for
% show why this implies 
the typical radius of a query -- which we will assume to be the distance to the nearest neighbour of query centre.
% -- also behaves nicely.

\begin{lemma}[M. Gromov, V.D. Milman]\cite{gromov:83}\label{l:m}
Let $(\Omega,\rho,\mu)$ denote a metric space with measure and $\alpha$ its concentration function.
Then if $A\subset\Omega$ is such that $\mu(A)>\alpha(\gamma)$ for some $\gamma>0$, it implies that $\mu(A_{\gamma})>1/2$.
\end{lemma}
% \begin{proof}
% Assume not and let $B=A_{\gamma}^c$.
% Then $\mu(B)\geqslant 0$, which implies $\mu(B_{\gamma}^c)\leqslant \alpha(\gamma)$.
% But $\mu(A)\leqslant \mu(B_{\gamma}^c)$ , a contradiction.
% \end{proof}

\begin{theorem}
Let $(\Omega_d,\rho_d,\mu_d,X_d)_{d=1}^{\infty}$ be a sequence of metric spaces with measure, forming a L\'evy family, together with i.i.d. samples $X_d$.
% from $\Omega_d$.
% \newline 
Assume that $n=n_d=|X_d|=d^{o(1)}$.
Furthermore, if $M_d$ denotes the median value of $\left\{ \rho_d(\omega_1,\omega_2)\vert \omega_i\in\Omega_d \right\}$, we assume that 
$M_d=\Theta(1)$, that is, for some fixed $c_1,c_2>0$, $\quad\forall d, c_1<M_d<c_2$.

Let $\rho_d^{(NN)}(\omega)$ denote the distance to the nearest neighbour of $\omega\in\Omega_d$ in $X_d$.
Define $m_d$ to be the median of $\rho_d^{(NN)}(\omega)$.
Then there exists some $c_3>0$ and some $D$ such that $\forall d\geqslant D$, $m_d>c_3$.
\end{theorem}
\begin{proof}
Assume the conclusion fails, then without loss of generality and proceeding to subsequence if necessary, $m_d\rightarrow 0$.
By definition of $m_d$, we know that for any $d$,
\[\mu_d\left( \bigcup_{x\in X}{B_{m_d}(x)}\right)\geqslant \frac{1}{2}.\]
It follows that 
\[n_d \sup_{\omega\in\Omega_d}{\mu_d \left( B_{m_d}(\omega)\right)}\geqslant \frac{1}{2},\]
and so we can find for any $d$ a point $\omega_d\in\Omega_d$ such that
\[\mu_d\left( B_{m_d}(\omega_d)\right)\geqslant \frac{1}{2n}.\]
If we denote by $\alpha_d$ the concentration functions of our spaces $\Omega_d$ we know by assumption the existence of $C,c>0$ s.t.
\[\forall d, \alpha_d(\epsilon)\leqslant C\text{e}^{-c\epsilon^2d}=d^{-o(1)}.\]
Hence we can find $d'$ s.t. $\alpha_{d'}(\gamma) <1/{2n_{d'}}$ and $m_{d'}<{c_1}/8$, where $\gamma = c_1/8$ as well.
% this is since eventually,
% \[C\text{e}^{-c\gamma^2d} < \frac{1}{2n_d} =d^{-o(1)}.\]
By lemma \ref{l:m}
\[\mu_{d'}\left( B_{m_{d'}}(\omega_{d'})\right)_{\gamma}\geqslant \frac{1}{2}.\]
It then follows that
\[\mu_{d'}\left(\left( B_{m_{d'}}(\omega_{d'})\right)_{\gamma}\right)_{\gamma}\geqslant 1-C\text{e}^{-c\gamma^2d'} \]
that is, since $m_{d'}+2\gamma<3c/8$,
\[\mu_{d'}\left(B_{3c_1/8}(\omega_{d'})\right)\geqslant 1-C\text{e}^{-c\gamma^2d'}. \]
But $\text{diameter}\left(B_{3c_1/8}(\omega_{d'})\right)\leqslant 3c_1/4$, so in $\Omega_{d'}\times\Omega_{d'}$ the measure of the set of points $(\omega_1,\omega_2)$ for which $\rho_{d'}(\omega_1,\omega_2)<c_1$ is at least
\[\left(1-\frac{1}{2n_{d'}}\right)^2,\]
obviously contradicting $M_{d'}>c_1$.
\end{proof}

This result frees us from having to consider a radius that vanishes as $n$, $d$ go to infinity.
With this achieved, let us recap our goal: to show that a large proportion of queries are slow, something along the lines of:
\begin{equation}\label{eq:z2}
\mathrm{median}_{q,p_i,r}\left(\mu_{n}( C_{q,p_1\ldots p_{k(n)},r(n)})\right)\stackrel{P}{\longrightarrow} 0 \text{ as } n,d\longrightarrow\infty,
\end{equation}
where the median is taken over all the queries under consideration: any $q\in\Omega_d$ and any $r$ at least as large as the distance to the nearest neighbour of $q$ in $X$.
As well, for each $d$ and $n=n_d$ we would like to also consider all possible pivot-based index schemes (as long as $k$ is within certain ranges we will specify later).
% 
% Why the median? The aim is to show a certain behaviour for {\em many} queries: at least half is dramatic enough.
So far we have shown, although the proof was just sketched (and with the detail about $k$ left out) that 
\begin{equation}\label{eq:z1}
\mathrm{median}_{q,p_i,r}\left(\mu( \mathcal{C}_{q,p_1\ldots p_{k(n)},r(n)})\right)\longrightarrow 0 \text{ as } n,d\longrightarrow\infty
\end{equation}
% Which is fine as long as $X=\Omega_d$ and hence the selection of a small dataset from an underlying large space is not taken into account.
% A more likely situation however is of a finite $X=X_d$ and an infinite (or at least much larger) $\Omega=\Omega_d$.
What we need is to find out when (\ref{eq:z1}) implies (\ref{eq:z2}).
To do so we will summon the powerful machinery of statistical learning theory.
%--------------------------------------------------------------------------------------------------------------
\section{Statistical learning theory}
Statistical learning theory has already been used in the analysis and design of indexing algorithms \cite{kleinberg} and is a vast subject.
We will just focus on the generalization of the Glivenko-Cantelli theorem due to Vapnik and Chervonenkis.

%\subsection{VC dimension}
\begin{theorem}[Glivenko-Cantelli]
Given sample $(X) = X_1,X_2\ldots X_n$ distributed i.i.d. according to {\em any }measure $\mu$ on $\mathbb{R}^n$, we have:
\[\sup_{r\in\mathbb{R}} \left\vert \mu_n (-\infty , r] - \mu (-\infty , r]\right\vert \stackrel{P}{\longrightarrow} 0. \]
\end{theorem}

\noindent
%The convergence is taken with respect the product measure induced by the sample.
%This theorem provides a means of linking the empirical distribution
%\[F_n(r):=\mu_n (-\infty , r]\]
%and the actual distribution
%\[F(r):=\mu (-\infty , r].\]

We can see this statement in terms of the empirical measures of particular subsets converging to their true measure.
This is made clear when we restate the theorem as follows: 
\begin{equation}
\sup_{A \in\mathcal{A}} |\ \mu_n (A) - \mu (A)\ | \stackrel{P}{\longrightarrow} 0, \label{close1}
\end{equation}
where
\[\mathcal{A} = \{ (-\infty , r] | r\in\mathbb{R}\},\]
which makes more apparent a path for extension: to generalize to other collections of subsets $\mathcal{A}$.
% , on general metric spaces.% $(\Omega,\mu)$ other than the real line.

%The bulk of the work is in finding appropriate measures of ``size'' of collections that can determine if \eqref{close1} takes place, and if so at what rate.

% To that purpose, we connect the potentially infinite collection $\mathcal{A}$ with the finite sample:

% Given $(X) = X_1,X_2\ldots X_n$ a 
A collection $\mathcal{A}$ ``colours'' the sample $X$ as follows. Each $A\in\mathcal{A}$ will assign 1 to $X_i$ if $X_i\in A$, and $0$ otherwise.
% otherwise it assigns 0. Hence we get a colouring of type
%  \[0,1,0,0,1,0\] 
%  which is an n-length encoding that might as well have been 
%\begin{center}  
%  white-black-white-white-black-white
%\end{center} 
We denote by $N(X)$ the {\em number} of such different colourings of $X$ generated by all $A\in\mathcal{A}$.
Clearly $N(X)\leqslant 2^n$. What is surprising is that in many situations, despite a seemingly rich $\mathcal{A}$, we have $N(X)\ll 2^n$. 

% 
% \begin{definition}
% The {\em (random) entropy} of sample $X=X^{(n)}$ is just $\ln(N(X))$, denoted by $H(X)$. 
% It follows that $H(X)\leqslant n\ln 2$.
% The expected value of $H(X)$, w.r.t. the sample distribution (in effect a product of $\mu$'s) is called the {\em entropy} of size n, denoted by $H(n)$:
% \[H(n) = \mathrm{E}\left(H(X^{(n)}\right)\]
% \end{definition}
% 
% A result in \cite{vapnik} states that \eqref{close1} is equivalent to
 % \[\frac{H(n)}{n}\stackrel{n\rightarrow\infty}{\longrightarrow} 0\]
% 
% This however is of little use if $\mu$ is unknown and does not guarantee {\em fast} (i.e. exponential) convergence.

\begin{definition}
The {\em growth function} $G=G_{\mathcal{A}}$ of a family $\mathcal{A}$ is defined by
\[G(n) = \ln\ \sup_{|X|=n}N(X).\]
It is independent of $\mu$ and the choice of sample $X$. 
\end{definition}
There are two cases to consider for an upper bound for the growth function \cite{vapnik}:
\begin{itemize}
\item  for all $n$, $G(n)=n\ln 2$
\item or, for the largest $\Delta$ such that $G(\Delta)=\Delta\ln 2$,
\begin{equation*}
G(n)\left\{\begin{array}{rrl}
=& n\ln 2 & \text{if } n\leqslant \Delta\\
\leqslant &\Delta (1+\ln (n/\Delta)) & \text{if } n>\Delta\\
\end{array} \right.
\end{equation*}
\end{itemize}
%Which says that $G$ can be either linear in $n$ or logarithmic after some point $\Delta$.
This $\Delta$ is the so-called {\em VC dimension} and it turns out that its finiteness is
a necessary and sufficient condition for 
\eqref{close1}. 
% Terminology in the literature also refers to {\em infinite} VC dimension when a finite $\Delta$ does not exist.
The rate of convergence 
% in the case of any $\mathcal{A}$ of finite VC dimension $\Delta$ 
is as follows (\cite{vapnik} p.148):

\begin{theorem}\label{thm:conv}[Vapnik--Chervonenkis]\label{thm:gkg}
For a collection $\mathcal{A}$ of subsets of $\Omega$, of finite VC dimension $\Delta$, and any measure $\mu$ on $\Omega$, we have that for any $\varepsilon>0$,
\begin{align*}
 P&\left[ \sup_{A \in\mathcal{A}} |\ \mu_n (A) - \mu (A)\ |>\varepsilon \right]< \\
  &4\exp \left[\left(\frac{\Delta (1+\ln (2n/\Delta))}{n} -\left(\varepsilon - \frac{1}{n}\right)^2\right) n \right].\\
\end{align*}
\end{theorem}
The convergence is eventually like $\exp (-\varepsilon^2 n)$, which is again a fast rate of convergence. 
% A somewhat different form for the right side is shown in \cite{devroye}:
% %\begin{align*}
% %P&\left[ \sup_{A \in\mathcal{A}} |\ \mu_n (A) - \mu (A)\ |>\varepsilon \right]<\\
% % &
% \[8\exp \left(\Delta (1+\ln{n / \Delta})\right) \exp\left(\frac{-n\varepsilon^2}{32}\right).\]
%\\
%\end{align*}
% Except the assumption that $\mathcal{A}$ is of finite VC dimension no other information is used.
Since no information about the measure $\mu$ is incorporated, the left side can be replaced by its supremum taken over all possible probability measures on the domain $\Omega$.

% Depending on the specific case, tighter bounds may be possible, using other capacity concepts than the VC dimension and a priori knowledge about the measure $\mu$ \cite{vapnik}, \cite{mendelson:03}.

A natural restatement of these results is to ask how big does the sample size $n$ have to be for the expression on the left to be less than some $\eta >0$. 
Solving for $\eta$ and the use of some technical inequalities (cf e.g. \cite{mendelson:03}) yields:
\begin{equation}
\label{eq:n}
n\geqslant \frac{128}{\varepsilon ^2}\left(\Delta\log \frac{2\text{e}^2}{\varepsilon} +\log\frac{8}{\eta}\right). 
\end{equation}

Calculations of VC dimension have been done for various objects (e.g. \cite{dudley:84}, \cite{vapnik}, \cite{devroye}):
The VC dimension of half-spaces $\{ x\in\mathbb{R}^d | (x,v)\geqslant b\}$ in $\mathbb{R}^d$ is $d+1$. 
The VC dimension of all open (or closed) balls in $\mathbb{R}^d$
%\[\left\{\{ x\in\mathbb{R}^d |\ \|x-v\| < r\}\quad ,\text{ where } v\in\mathbb{R}^d \text{, }r\in\mathbb{R}\right\}\]
is also $d+1$.
Axis-aligned rectangular parallelepipeds in $\mathbb{R}^d$, i.e. sets of form
\[ [a_1,b_1]\times [a_2,b_2]\times\ldots\times [a_d,b_d] \]
have a VC dimension of $2d$. 

Our interest is in calculating the VC dimension of all possible set of form $\mathcal{C}_q$, the collection of which for a fixed $k$ we denote:
\begin{equation}
\label{eq:A}
\mathcal{A}=\mathcal{A}_k=\{\mathcal{C}_{q,p_1\ldots p_{k(n)},r(n)}|q\in\Omega , p_i\in\Omega, r>0\}
\end{equation}

As
\begin{align*}
 \mathcal{C}&= \{\omega : \sup_i |\ \| \omega-p_i\| - \| q-p_i\| \ | > r \}\\
 						& = \big( \bigcap_i \{\omega : \ |\ \| \omega-p_i\| - \| q-p_i\| \ | \leqslant r \}\big)^c, 
\end{align*}
we can proceed through several steps.
%\begin{itemize}
%\item
A set of the form 
\[\{\omega : \ |\ \| \omega-p_i\| - \| q-p_i\| \ | \leqslant r \}\]
is a ``spherical shell,'' and
% i.e. the interesection of one ball with the complement of a smaller ball having the same center. 
% We note that 
an intersection of shells is an interesection of sets from $\mathcal{A}\cup\mathcal{A}^c$, where $\mathcal{A}$ is the collection of all balls.
%\item
It is easy to show that given a collection $\mathcal{A}$ the complement collection 
$\mathcal{A}^c=\{A^c | A\in\mathcal{A}\}$
has the same VC dimension.
%For assume $\mathcal{A}$ shatters $X$ and take any colouring of $X$. The opposite colouring, by putting 1 instead of 0 and 0 instead of 1,
%is produced by some $A\in\mathcal{A}$. Then $A^c\in\mathcal{A}^c$ produces the original colouring. The same argument can be applied in the other direction.
%\item
The VC dimension of balls was quoted above as $d+1$, hence the VC dimension of complements of balls is $d+1$ as well. 
The VC dimension of the {\em union} of the two collections is \[(d+1) + (d+1) +1 = 2d + 3,\]
%\end{itemize}
as a consequence of a general result \cite{vidyasagar:03}:
% \begin{lemma}
If a collection $\mathcal{A}$ has VC dimension $\Delta_a$ and a collection $\mathcal{B}$ has VC dimension $\Delta_b$, the union $\mathcal{A}\cup\mathcal{B}$
has VC dimension at most $\Delta_a+\Delta_b+1$.
%\end{lemma}

A result for intersection of sets is mentioned in \cite{blumer}:
\begin{lemma}
For $(\Omega,\rho)=(\mathbb{R}^d,L^2)$, an upper bound on the VC dimension of $\mathcal{A}_{\cap_k}$, composed of $k$-fold interesections of elements of a family $\mathcal{A}$ of VC dimension $\Delta$ is $2\Delta k\ln (3 k)$.
\end{lemma}

Hence we can conclude that the VC dimension of $\mathcal{A}_k$ for the case $\Omega\subset\mathbb{R}^n$ is bounded by
\begin{equation}\label{eq:vc}
2(2d+3)(2k)\ln ((3)(2k)) = k(8d+12)\ln (6k),
\end{equation}
where $k$ is the number of pivots. 
% This also allows us to conclude that convergence like in equation \eqref{close1} takes place.
%Of course we only considered the case of $\mathbb{R}^d$ with the normal Euclidian metric.

% We have already mentioned that the VC dimension of axis-aligned ``boxes'' is $2d$.
% It so happens that all balls with respect the $L^{\infty}$ metric are such boxes, so we can obtain a bound on $\mathcal{A}_k$ for this metric as well.

Another example comes from considering the Hamming cube. As there 
% In a similar argument to one we presented above, we observe that 
are $2^d$ points in a $d$-dimensional Hamming cube, and at most $d$ different radii, so at most $d2^d$ different balls exist.
We know from e.g. \cite{blumer} that
% an upper bound on the VC dimension of finite collections:
% \begin{lemma}[Finite $\mathcal{A}$]
% If 
if the class $\mathcal{A}$ is finite, its VC dimension is bounded by $\log_2\vert\mathcal{A}\vert$.
%\end{lemma}
Disregarding the small leftover term, the VC dimension for balls in the Hamming cube is about $d$.

Summarizing:

\begin{theorem}\label{thm:deltabounds}
Let us denote by $\Delta$ the VC dimension of collection $\mathcal{A}_k$ as defined in equation \eqref{eq:A}.
Then upper bounds on $\Delta$, depending on the metric space, are as follows:
\begin{itemize}
\item
For $(\mathbb{R}^d, L^2)$, $\Delta\leqslant k(8d+12)\ln (6k)$.
\item
For $(\mathbb{R}^d, L^\infty)$, $\Delta\leqslant k(16d+4)\ln (6k)$.
\item
For $(\Sigma^d,\rho)$, $\Delta\leqslant k(8d+8\log_2{d}+4)\ln (6k)$.
\end{itemize}
\end{theorem}

\section{Main result\label{s:main}}
\begin{theorem}
Consider a sequence of metric spaces $(\Omega_d,\rho_d)$, where $d=1,2,3,\ldots$ and the VC dimension of closed balls in $(\Omega_d,\rho_d)$ is $O(d)$.
Assume every $\Omega_d$ supports a Borel probability measure $\mu_d$ so that for some $C,c>0$ the concentration functions $\alpha_d$ of $(\Omega_d,\rho_d,\mu_d)$ satisfy \[\forall\epsilon>0,\quad\alpha_d(\epsilon)\leqslant C\text{e}^{-c\epsilon^2d}.\]
Select for each $d$ an i.i.d. sample $X_d$ of size $n_d$ from $\Omega_d$, according to $\mu_d$, where the sample size $n_d$ satisfies $d=\omega(\log n_d)$ and $d=n_d^{o(1)}$.
% if $d$ is expressed as a function of $n_d$.
Suppose further for every $d$ a pivot index for similarity search is built using
% For every $d$, build a pivot index for similarity search 
% We treat $X_d$ as a dataset on which to build an index for similarity search.
% The index built is a pivot index using 
$k$ pivots, where \[k=o(n_d/d).\]
Fix arbitrarily small $\varepsilon,\eta >0$.
Suppose we only ask queries whose radius is equal or greater to the distance to nearest neighbour of query centre $q\in\Omega_d$ in $X_d$.
\newline\bfseries 
Then there exists a $D$ such that for all $d\geqslant D$, the probability that {\em at least half} the queries on dataset $X_d$ take less than $(1-\varepsilon)n_d$ time is less than $\eta$.

Furthermore, if we allow the likelihood $\eta$ to depend on $d$, we can pick $\eta_d$ so that the above holds true and
\[\lim_{D\rightarrow\infty}{\prod_{d=D}^{\infty}(1-\eta_d)} = 1.\]
We emphasize that this result is independent of the selection of pivots.
\end{theorem}

\begin{proof}[Sketch of a proof]
From Eq. (\ref{eq:conc}) we know that, for a vast majority of query centres $q$,
\[\mu (\mathcal{C}_q) \leqslant M_2k\text{e}^{-dr^2},\]
where $M_2$ is some constant.

We will sacrifice a certain number of sets of form $\mathcal{C}_q$ so that $r$ can be considered a constant (see section \ref{sec:radius}): we will proceed with at least half the queries having radius $r$ above a constant independent of $d$.
Hence the quantities that vary in $d$ are $n$ and $k$.
Since $d$ is superlogarithmic in $n$,
\begin{equation*}
\begin{array}{ll}
&\forall c>0 \text{, } d>c\log n\\
\Rightarrow& \forall c>0 \text{, } \exp(-d) <\exp(-c\log n)\\
\Rightarrow&\forall c>0 \text{, }\exp(-d)<cn.
\end{array}
\end{equation*}

So $\text{e}^{-dr^2}=o(n)$, and hence $\mu (\mathcal{C}_q) = o(n)$.
%so not only does $\mu(\mathcal{C})\rightarrow 0$, we have a bound on the convergence as well.
In fact this holds for at least half the queries $q$ simultaneously, so:
\[\mathrm{median}\sup_{\mathcal{C}_q}\mu (\mathcal{C}_q)=o(n).\]
From the previous section, we know that only for large values of $n$ will empirical measures be close (up to $\varepsilon$) to actual measures with likelihood (1-$\eta$ ).
The lower bound on $n$ then naturally depends on $\varepsilon$, $\eta$ but also on the VC dimension $\Delta$ of the collection $\mathcal{A}_k$.

Let us fix $\varepsilon=1/2$ and assume $\eta$ is bounded by some value less than 1.
Then by pooling all constants, including $\varepsilon$ and $\eta$ but not $\Delta$ we can rewrite expression \eqref{eq:n} as:
\begin{equation}\label{eq:nsimple}n\geqslant M_1 \Delta,\end{equation}
where $M_1>1$.
What we would like to avoid is to have the right part of this expression grow linearly in $n$.
We know an upper bound on $\Delta$ depends on $k$ and $d$ as established in Theorem \ref{thm:deltabounds}. As our concern is for asymptotic behaviour we will simplify this bound to $ k d \ln k$ .

% We will generalize from our examples of L\'evy families to define (an) acceptable asymptotic behaviour for the VC dimension of balls: linear in $d$.
Combining $d=o(n)$ with the asymptotic condition on $k$, we conclude that:
\[\Delta=o(n),\] 
and hence asymptotically we know that the right side of expression \eqref{eq:nsimple} falls (much) under $n$.
Therefore we are able to conclude:
\begin{equation}\label{eq:close2}
P(\sup_{\mathcal{C}_q}\vert \mu_{\#}(C_q)-\mu(\mathcal{C}_q)\vert > \varepsilon)<\eta,
\end{equation}
which, combined with $\varepsilon=1/2$ and $\mathrm{median}\sup_{\mathcal{C}_q}\mu (\mathcal{C}_q)=o(n)$, gives the first part of the result.

According to expression
% To derive the asymptotic probability, expression 
(\ref{eq:n}),
% can be used to obtain the lower bound on $\eta$:

\[\eta\geqslant \exp \left( \Delta\log \left(\frac{2\text{e}^2}{\varepsilon}\right)+\log8-\frac{\varepsilon^2 n}{128}\right)
%\]
%
%which as an asymptotic function of $d$ transforms into
%\[\eta
=\exp (-d^{\omega(1)}).\]

Assuming independent choices of the datasets $X_d$, and assuming that for each $d$ the probability of an event is at least $1-\eta_d$, we aim to prove that
\[\lim_{D\rightarrow\infty}{\prod_{d=D}^{\infty}(1-\eta_d)} = 1.\]
%
% \[\prod_{d=D}^{\infty}(1-\eta_d)\]
%
As $\eta_d$ goes to 0 at least as fast as $\text{e}^{-d}$, it is enough to show that
\begin{equation}
\label{eq:enough}
\lim_{D\rightarrow\infty}{\prod_{d=D}^{\infty}(1-\text{e}^{-d})} = 1.\end{equation}
%
% to have 
%
% \[\lim_{D\rightarrow\infty}{\prod_{d=D}^{\infty}(1-\eta_d)} = 1\]
%
% as well.

%From \cite{ash} the convergence of $\sum_{d=D}^{\infty}{\text{e}^{-d}}$ is sufficient to conclude the convergence of $\prod_{d=D}^{\infty}(1-\text{e}^{-d})$. 
Observing \cite{ash} that for any sequence $0\leqslant\eta_d\leqslant 1$,
\[1-\sum_{d=1}^{N}{\eta_d}\leqslant \prod_{d=1}^{N}(1-\eta_d)\leqslant \exp\left(\sum_{d=1}^{N}{-\eta_d}\right),\]
we can extend this, for any $D$ to:
\[1-\sum_{d=D}^{\infty}{\eta_d}\leqslant \prod_{d=D}^{\infty}(1-\eta_d)\leqslant \exp\left(\sum_{d=D}^{\infty}{-\eta_d}\right).\]
Summing the geometric series, we obtain Eq. (\ref{eq:enough}).
% We know that as a geometric series,
% \[\sum_{d=D}^{\infty}{\text{e}^{-d}}=\frac{\text{e}^{-d}}{1-\text{e}}\]
% Hence we can conclude that 
% \[\lim_{D\rightarrow\infty}{\prod_{d=D}^{\infty}(1-\text{e}^{-d})} = 1.\]
\end{proof}

%-------------------------------------------------
\subsection{Conclusion}
We have established a rigorous asymptotically linear lower bound on the expected average performance of the optimal pivot-based indexing schemes for similarity search in datasets randomly sampled from domains whose dimension goes to infinity. 
The examples given above of the various spaces exhibiting normal concentration of measure should convince the reader that many of the most naturally occuring domains and measure distributions are such.

% it is real and widespread, though of course not universal.
%For the VC dimension of balls to depend linearly on $d$ is a more vague requirement as the definition of dimension for metric spaces for the purposes of similarity search is an unresolved problem.
%It is however clearly impossible to take out dimension from a discussion of the curse of dimensionality, which is after all what we have shown, caveats notwhithstanding.

This is not the first lower bound result for pivoting algorithms for exact similarity search.
% There has been research for some time into lower bounds for various cases, though most often for approximate algorithms e.g. \cite{chakrabarti}. 
A specific lower bound for pivot-based indexing already mentioned above is that of \cite{chavez:2}:
\[\tilde{d}\log n.\]
This result assumes that $k=\Theta(\log n)$.
Furthermore and more importantly, the pivot selection is assumed to be random, as opposed to our (much stronger) bound that is applicable to {\em any} pivot selection technique.

Other, more general asymptotic analyses considering more classes of indexing schemes \cite{weber:98, shaft:06} fix $n$ or in the case of \cite{weber:98} also fail to distinguish between the dataset and the dataspace making results appear stronger than they actually are. 

The aim in \cite{shaft:06} was to demonstrate that
\[\frac{\mathrm{E}(\text{cost})}{n}\stackrel{n}{\longrightarrow} 1\]
which came at the expense of any results on the rate of convergence. We chose instead to prove a weaker result, with convergence to some number close to 1/2 but with estimates on the rate of convergence.

It should be assumed that the hypotheses of our paper are universal. Rather, our theoretical analysis confirms that at least in some settings, the curse of dimensionality for pivot-based schemes is indeed in the nature of data.
Probably a more realistic situation from the viewpoint of applications would be that of an intrinsically low dimensional dataset contained in a high-dimensional domain, and performing an asymptotic analysis of various indexing schemes in this setting is an interesting open problem.

%If we are to apply our restrictions on $d$ (assumed to be asymptotically equivalent to $\tilde{d}$), the result is that
%\[\tilde{d}\log n = \omega(\log^2 n)\]
%\[\tilde{d}\log n=n^{o(1)}\]
%So if we are to use these results asymptotically they do not provide strong lower bounds on the cost of similarity search.

\bibliographystyle{alpha}

\begin{thebibliography}{}

\bibitem[Ash71]{ash}
Ash, R.B. (1971)
\newblock
Complex Variables.
\newblock
Academic Press.

% \bibitem[Ben00]{benyamini}
% Benyamini, Y., Lindenstrauss, J. (2000)
% \newblock
% Geometric Nonlinear Functional Analysis: V. 1
% \newblock
% AMS Colloquium publications 

\bibitem[Bey99]{beyer}
Beyer, K, Goldstein, J., Ramakrishnan, R., Shaft, U. (1999)
\newblock
When is ``Nearest Neighbour'' meaningful?
\newblock
Lect. Notes in Comp. Sci., vol. 1540, 217--235

\bibitem[Blu89]{blumer}
Blumer, A., Ehrenfeucht, A., Haussler, D., Warmuth, M.K. (1989)
\newblock
Learnability and the Vapnik-Chervonenkis dimension.
\newblock
{\em Journal of the ACM}, 36, 
% Issue 4, 
929--965.

\bibitem[Bus03]{bustos}
Bustos, B., Ch\'avez, E., Navarro, G (2003)
\newblock
Pivot selection techniques for proximity searching in metric spaces.
\newblock
Pattern Recognition Letters, 24-14, 2357--2366.

\bibitem[Cia97]{ciaccia:97}
Ciaccia, P., Patella, M., Zezula, P. (1997)
\newblock
M-tree: An efficient access method for similarity search in metric spaces.
\newblock
Proc. VLDB 1997: pp426-435

\bibitem[Cia98]{ciaccia:98}
Ciaccia, P., Patella, M., Zezula, P. (1998)
\newblock
A cost model for similarity queries in metric spaces.
\newblock
Proc. 17th ACM SIGACT-SIGMOD-SIGART Symposium on Principles of Database Systems.

% \bibitem[Chak]{chakrabarti}
% Chakrabarti, A., Chazelle, B., Gum, B., Lvov, A. (1999) 
% \newblock
% A lower bound on the complexity of approximate nearest-neighbor searching on the Hamming cube.
% \newblock
% Proc. 31th Annual ACM Symposium on Theory of Computing, pp 305--311

\bibitem[Cha01] {chavez:1}
Ch\'avez, E., Navarro, G. Baeza-Yates, R. \& Marroqu\'\i n, J. L. (2001).
\newblock
Searching in metric spaces.
\newblock ACM Computing Surveys, 33, 273--321.

\bibitem[Cha01b] {chavez:2}
Ch\'avez, E., Navarro, G. (2001).
\newblock
Towards measuring the searching complexity of metric spaces.
\newblock Proceedings of ENC'01, 969--978.

%\bibitem[Cha03]{chavez:3}
%Ch\'avez, E., Navarro, G. (2003).
%\newblock
%Probabilistic proximity search: Fighting the curse of dimensionality in metric spaces.
%\newblock
%Information Processing Letters 85 pp. 39--46

\bibitem[Cha05]{chavez:5}
Ch\'avez, E., Navarro, G. (2005).
\newblock
A compact space decomposition for effective metric indexing.
\newblock
Pattern Recognition Letters, v.26, issue 9, pp 1363 -- 1376 

\bibitem[Cla05]{clarkson:05}
Clarkson, K.L. (2005)
\newblock
Nearest-Neighbor Searching and Metric Space Dimensions.
\newblock
Nearest-Neighbor Methods for Learning and Vision: Theory and Practice, MIT Press, 2006.

\bibitem[Dev97]{devroye}
Devroye, L., Gy\"orfi, L., Lugosi, G. (1997)
\newblock
A Probabilistic Theory of Pattern Recognition.
\newblock
Springer

\bibitem[Dud84]{dudley:84}
Dudley, R.M. (1984).
\newblock
A Course on Empirical Processes, Lecture Notes in Math.
\newblock
Springer, New York, 1984

\bibitem[Gro83]{gromov:83}
Gromov, M., Milman, V.D. (1983).
\newblock
 A topological application of the isoperimetric inequality.
\newblock
Amer. J. Math. 105, 843--854.

\bibitem[Heg01]{hegland:1}
Hegland, M. (2001)
\newblock
Data mining techniques.
\newblock
Acta Numerica, 2001 pp313-355

\bibitem[Ind04]{indyk:1}
 Indyk, P. (2004).
\newblock
Ch. 39 of Handbook of Discrete and Computational Geometry, Goodman, J.E., O'Rourke, J., eds.,
\newblock
CRC Press.

\bibitem[Kle97]{kleinberg}
Kleinberg, J. (1997)
\newblock
Two algorithms for nearest-neighbor search in high dimensions.
\newblock
Proc. 29th ACM Symposium on Theory of Computing.

% \bibitem[Led01]{L}
% Ledoux, M. (2001). 
% \newblock 
% The concentration of measure phenomenon. 
% \newblock
% Math. Surveys and 
% Monographs {\bf 89}, Providence: Amer. Math. Soc.

\bibitem[Men03]{mendelson:03}
Mendelson, S. (2003)
\newblock
A few notes on statistical learning theory.
\newblock
Advanced Lectures in Machine Learning LNCS 2600, 1-40.
\newblock
Springer

\bibitem[Mil86]{milman:86}
Milman, V.D., Schechtman, G. (1986)
\newblock
Asymptotic theory of finite dimensional normed spaces.
\newblock
 Lecture Notes in Mathematics.

\bibitem [Pes00]{pestov:00}
Pestov, V. (2000). 
\newblock On the geometry of similarity search: dimensionality curse and concentration of measure.
\newblock {\em Inform. Process. Lett.} 73, 47--51.

% \bibitem [Pes02]{pestov:02}
% Pestov, V. (2002). 
% \newblock Elements of Asymptotic Geometric Analysis.
% \newblock Lecture notes, Victoria University of Wellington.

\bibitem [Pes08]{pestov:08} 
Pestov, V. (2008).
\newblock An axiomatic approach to intrinsic dimension of a dataset. 
\newblock Neural Networks 21, 204--213. 
%\newblock In: Proc. of the 22-nd Int. Joint Conf. on Neural Networks (IJCNN'07), Orlando, FL (pp. 1775--1780). 

%bibitem[Pes07b]{pestov:07b}
%Pestov, V. (2007)
%\newblock Elements of statistical machine learning: a mathematical introduction.
%\newblock Lecture notes, University of Ottawa.

\bibitem[Sha06]{shaft:06}
Shaft, U., Ramakrishnan, R. (2006)
\newblock
Theory of nearest neighbors indexability .
\newblock
ACM Transactions on Database Systems, Volume 31 , Issue 3, pp 814 -- 838.

\bibitem[Vap98]{vapnik}
Vapnik, V. (1998)
\newblock
Statistical Learning Theory.
\newblock
Wiley series on adaptive and learning systems for signal processing, communications and control.

\bibitem[Vid03]{vidyasagar:03}
Vidyasagar, M. (2003)
\newblock
Learning and Generalisation: With Applications to Neural Networks.
\newblock
Springer.

\bibitem[Web98]{weber:98}
Weber, R., Schek, Hans-J., Blott, S. (1998)
\newblock
A Quantitative Analysis and Performance Study for Similarity-Search Methods in High-Dimensional Spaces 
\newblock
Proceedings of the 24rd International Conference on Very Large Data Bases, pp194-205.

\bibitem[Zez05]{zezula}
Zezula, P., Amato, G., Dohnal, V., Batko, M. (2005)
\newblock
Similarity search: the metric space approach.
\newblock
{\em Springer series: advances in database systems}, vol 32.


\end{thebibliography}

\end{document}